\documentclass[10pt, doublecolumn]{IEEEtran}
\usepackage{epsfig,latexsym}
\usepackage{float}
\usepackage{indentfirst}
\usepackage{amsmath}
\usepackage{bm}
\usepackage{amssymb}
\usepackage{times}

\usepackage{algorithm}
\usepackage[noend]{algpseudocode}
\usepackage{subfigure}
\usepackage{psfrag}
\usepackage{hyperref}
\usepackage{cite}
\usepackage{lastpage}
\usepackage{fancyhdr}
\usepackage{color}
 \usepackage{amsthm}
\usepackage{bigints}
\newtheorem{Lemma}{Lemma}
\newtheorem{lemma}[Lemma]{$\mathbf{Lemma}$}

\newcounter{problem}
\newcounter{save@equation}
\newcounter{save@problem}
\makeatletter
\newenvironment{problem}
{\setcounter{problem}{\value{save@problem}}%
  \setcounter{save@equation}{\value{equation}}%
  \let\c@equation\c@problem
  \subequations
}
{\endsubequations
  \setcounter{save@problem}{\value{equation}}%
  \setcounter{equation}{\value{save@equation}}%
}

\begin{document}
\title{\vspace{-1em} \huge{  NOMA-Based Coexistence  of Near-Field and Far-Field Massive MIMO Communications  }}

\author{\author{ Zhiguo Ding, \IEEEmembership{Fellow, IEEE},   Robert Schober, \IEEEmembership{Fellow, IEEE}, and H. Vincent Poor, \IEEEmembership{Life Fellow, IEEE}   \thanks{ 
  
\vspace{-2em}

  Z. Ding is with  Khalifa University, Abu Dhabi, UAE, and  University of Manchester, Manchester, M1 9BB, UK.   R. Schober is with  
Friedrich-Alexander-University Erlangen-Nurnberg (FAU), Nuremberg, 91054, Germany. H. V. Poor is  with   Princeton University, Princeton, NJ 08544,
USA. 

Z. Ding's work was  supported by the UK EPSRC under grant number EP/W034522/1and H2020 H2020-MSCA-RISE-2020 under grant number 101006411.  R. Schober’s work was (partly) funded by the German Research Foundation (DFG) under project SFB 1483 (Project-ID 442419336 Empkins) and the BMBF under the program of ``Souverän. Digital. Vernetzt.'' joint project 6G-RIC (Project-ID 16KISK023). H. V. Poor's work was supported by the U.S National Science Foundation under Grant CNS-2128448. }
 
 \vspace{-2em}

  }\vspace{-5em}}
 \maketitle

\vspace{-1em}
\begin{abstract}
This letter considers  a legacy massive multiple-input multiple-output (MIMO) network, in which spatial beams have been preconfigured for   near-field users, and  proposes  to use  the non-orthogonal multiple access (NOMA)     principle to serve  additional  far-field users  by  exploiting the  spatial beams preconfigured for the legacy near-field users.  Our results reveal that     the coexistence between near-field and far-field communications can be effectively supported via NOMA, and that the performance of NOMA-assisted massive MIMO   can be efficiently    improved by increasing the number of   antennas at the base station.  
\end{abstract}\vspace{-0.1em}


\vspace{-1em}
\section{Introduction}
  One recent advance  in non-orthogonal multiple access (NOMA) is  its use   as an add-on in a massive  multiple-input multiple-output (MIMO) based legacy space division multiple access (SDMA) network, where   spatial beams preconfigured for legacy users are used to serve additional   users  \cite{9693536}. As a result,
the connectivity and the overall system throughput of SDMA
can be improved in a low-complexity and spectrally efficient
manner. For conventional     SDMA networks based on  far-field communications, where the transceiver distance is larger than the Rayleigh distance \cite{7400949}, this application  of NOMA   is intuitive, as explained in the following.  Far-field  beamforming is based on   steering vectors, i.e., the spatial beams are  cone-shaped \cite{7561012}. In practice, each of these cone-shaped beams  can cover a large area, and it is intuitive to encourage multiple users which are inside of one cone-shaped area  to exploit  the same beam via NOMA.

Near-field communications have     received a lot of attention as,  for high carrier frequencies  and   large numbers of antennas, the Rayleigh distance becomes significantly large \cite{9184098,devar}. Unlike far-field communications,    the spherical-wave channel model has to be used for   near-field communications, which motivates the use of   beam-focusing, i.e., a beam is focused    on  not only a spatial direction   but also a specific location \cite{9738442}. As a result, a naturally arising   question is whether the preconfigured spatial beams in near-field communication networks can still be used to admit additional users in the same manner as   far-field beams, which is the motivation for this work.

This letter considers  a legacy near-field  SDMA  network, in which spatial beams have been preconfigured for legacy near-field users, and  proposes  to apply the principle of NOMA  to serve additional  far-field users by  exploiting  these preconfigured spatial beams. A resource allocation optimization problem is   formulated  to maximize the far-field users' sum data rate while guaranteeing the legacy near-field users' quality-of-service (QoS) requirements. First, a suboptimal low-complexity algorithm based on successive convex approximation (SCA) is   proposed to solve the problem, and then the optimal performance   is obtained for two special cases by applying the branch-and-bound (BB) algorithm \cite{5765556, beamallocation}.      Simulation results are presented to demonstrate that the use of NOMA can effectively support the coexistence of   near-field and far-field communications,  and the performance of NOMA assisted massive MIMO   can be efficiently    improved by increasing the number of   antennas at the base station.

\section{System Model}
 Consider a legacy  downlink near-field SDMA network, in which a   base station employs an $N$-antenna uniform linear array (ULA) and serves   $M$ single-antenna near-field users, where $M\leq N$. In this letter, it is assumed that $M$ near-field beamforming vectors, denoted by $\mathbf{p}_m$, have already been configured to  serve the   legacy users individually. The aim of this letter is to admit $K$ additional  far-field users based on these preconfigured spatial beams.   
  Denote the $2$-dimensional    coordinates  of the $m$-th near-field user,   the $k$-th far-field user, the center of the array, and   the $n$-th element of the array by  ${\boldsymbol \psi}^{\rm NF}_m$, ${\boldsymbol \psi}^{\rm FF}_k$, ${\boldsymbol \psi}_0$, and ${\boldsymbol \psi}_n$, respectively. According to the principle of near-field communications, $|{\boldsymbol \psi}^{\rm NF}_m-{\boldsymbol \psi}_0|< d_{R}(N)$, and $|{\boldsymbol \psi}^{\rm FF}_k-{\boldsymbol \psi}_0|> d_{R}(N)$, where  $d_{R}(N)=\frac{2((N-1)d)^2}{\lambda}$, $\lambda$, and  $d$ denote  the Rayleigh distance,        the wavelength, and    the antenna spacing of the ULA, respectively  \cite{9738442,chengxiang,9482501}.  We note that other types of antenna arrays, such as uniform planar arrays (UPAs)  and  uniform circular  arrays (UCAs), can also be used for supporting  near-field   communications  \cite{9184098, ldaix1}.

 \vspace{-1em}
 \subsection{Near-Field and Far-Field Channel Models}
The $m$-th legacy near-field user's observation  is given by $
 y_m = \mathbf{h}_m^H \mathbf{x}+n_m$,
 where  $\mathbf{x}$ denotes the signal vector sent by the base station, $n_m$ denotes the additive Gaussian noise with its power denoted by $\sigma^2$, $\mathbf{h}_m$ is based on the spherical-wave propagation model  \cite{devar,9738442,Eldar2}:
 \begin{align}\label{near}
 \mathbf{h}_m = \alpha_m\begin{bmatrix} 
 e^{-j\frac{2\pi }{\lambda}\left| {\boldsymbol \psi}^{\rm NF}_m -{\boldsymbol \psi}_1\right|} &\cdots &  e^{-j\frac{2\pi }{\lambda}\left| {\boldsymbol \psi}^{\rm NF}_m -{\boldsymbol \psi}_N\right|}
 \end{bmatrix}^T,
 \end{align} 
where $\alpha_m = \frac{c}{4\pi f_c \left| {\boldsymbol \psi}^{\rm NF}_m -{\boldsymbol \psi}_0\right|}$, $c$, and $f_c$ denote the free-space path loss,   the speed of light, and    the carrier frequency, respectively.  We note that       the line of sight (LoS) path is assumed to be available for the near-field users, as they are within the Rayleigh distance from the base station \cite{9738442,Eldar2}.  
 
The $k$-th far-field user receives the following signal: $
 z_k = \mathbf{g}_k^H \mathbf{x}+w_k$,
where $w_k$ denotes the additive Gaussian noise having  the same  power as $n_m$,   the conventional beamsteering vector is  used to model the far-field user's channel vector, $\mathbf{g}_k$, as follows: \cite{7400949}
 \begin{align}\nonumber
 \mathbf{g}_k =& \alpha_k e^{-j\frac{2\pi}{\lambda}  \left| {\boldsymbol \psi}^{\rm FF}_k -{\boldsymbol \psi}_1\right| }\\\label{far} &\times \begin{bmatrix}
1& e^{-j\frac{2\pi d}{\lambda}\sin \theta_k} &\cdots &  e^{-j\frac{2\pi d}{\lambda} (N-1)\sin \theta_k}
 \end{bmatrix}^T,
 \end{align}
 and $\theta_k$ denotes the conventional angle of departure.  Scheduling is assumed to be carried out to ensure that each participating    far-field user has an LoS connection to the base station. Because the LoS link is typically  $20$ dB stronger than the non-LoS links, only the LoS link is considered in \eqref{near} and \eqref{far} \cite{7279196}. 
 
 {\it Remark 1:} Comparing \eqref{near} to \eqref{far}, one can observe that the near-field channel model is fundamentally different from the far-field one. In particular, the channel vector in \eqref{far} is mainly parameterized by the angle of departure,  $\theta_k$, but the elements of the vector in \eqref{near}  depend on  the near-field user's specific location.

  \vspace{-1em}
 \subsection{Near-Field Beamforming and NOMA  Data Rates}
 For   illustrative  purposes,    full-digital near-field beamforming   based on the zero-forcing principle is adopted in this letter: $
 \mathbf{P}\triangleq \begin{bmatrix}
 \mathbf{p}_1&\cdots &\mathbf{p}_M
 \end{bmatrix} = \mathbf{H}\left( \mathbf{H}^H\mathbf{H}\right)^{-1}\mathbf{Q}$,
 where $\mathbf{H}=\begin{bmatrix}
 \mathbf{h}_1&\cdots &\mathbf{h}_M
 \end{bmatrix}$, and $\mathbf{Q}$ is an $M\times M$ diagonal matrix to ensure power normalization. In particular, the $i$-th element on the main diagonal of $\mathbf{Q}$ is given by $\left[\mathbf{Q}\right]_{i,i} = \left[\left( \mathbf{H}^H\mathbf{H}\right)^{-1}  \right]_{i,i}^{-\frac{1}{2}}$,  which ensures that   the beamforming vectors are normalized, i.e.,  $\mathbf{p}_m^H\mathbf{p}_m=1$, for all $m\in\{1, \cdots, M\}$. 
 
The NOMA principle is applied to ensure that each preconfigured spatial beam  is used as a type of bandwidth resource for serving additional far-field users, which means that the signal vector sent by the base station is given by 
  \begin{align}
 \mathbf{x} = &\sum^{M}_{m=1}\mathbf{p}_m\left(\sqrt{P_m}s_m^{\rm NF}+ \sum^{K}_{k=1} f_{m,k} s^{\rm FF}_k \right), 
 \end{align}
 where $P_m$ is the  transmit power allocated to the $m$-th near-field user's signal, $f_{m,k}$ denotes the coefficient assigned to the $k$-th far-field user on beam $\mathbf{p}_m$, and $ {s}^{\rm NF}_m $ and $ {s}^{\rm FF}_k $ denote the signals for the near-field and far-field users, respectively.   If the $k$-th far-field user uses only a single beam, $f_{m,k}$ can be viewed as a power allocation coefficient. If multiple beams are used, $f_{m,k}$ can   be interpreted  as a beamforming coefficient. 
  
Therefore, the observation at the $m$-th near-field user can be written as follows:
  \begin{align}
 y_m = & 
\mathbf{h}^H_m \mathbf{p}_m\left(\sqrt{P_m}  {s}_m^{\rm NF} +  \sum^{K}_{k=1} f_{m,k} s^{\rm FF}_k \right)+n_m.
 \end{align}
For the considered NOMA network,   the near-field users have better channel conditions than the far-field users, and hence have the capability to carry out SIC.   To reduce the system complexity, assume that at most a single far-field user is scheduled on a near-field user's beam, and each of the $K$ far-field users utilize  $D_x$ beams, where $KD_x\leq M$. Denote   the subset collecting the indices of the beams used by the $k$-th far-field user by $\mathcal{S}_k$, where $|\mathcal{S}_k|=D_x$. 
 
Assume that   the $k$-th far-field user is active on $\mathbf{p}_m$. On the one hand, this far-field user's signal can be decoded by  the $m$-th near-field user with the following data rate: $
R_{m,k}^{\rm FF-NF} = \log\left(
1+ \frac{  |f_{m,k}|^2 h_m }{\sigma^2+P_mh_m}
\right)$,
where $h_m=|\mathbf{h}_m^H\mathbf{p}_m|^2$. 
If the first stage of SIC is successful, the near-field user can remove the far-field user's signal and decode its own signal with the following data rate: $
R_m^{\rm NF} = \log\left(
1+ \frac{P_m}{\sigma^2} h_m
\right)$.

On the other hand,    the $k$-th far-field user directly decodes its own signal with the following data rate: 
 \begin{align}\nonumber
 R^{\rm FF}_k = \log\left(1+
 | \tilde{\mathbf{g}}^H_k\mathbf{f}_k|^2 \gamma_k^{-1} 
  \right),
 \end{align}
 where $\gamma_k=\sigma^2+\sum^{M}_{m=1}P_m g_{m,k} +\underset{i\neq k}{\sum} | \tilde{\mathbf{g}}^H_k\mathbf{f}_i|^2$ $\tilde{\mathbf{g}}_k=\mathbf{P}^H\mathbf{g}_k$, $g_{m,k} = |\mathbf{g}_k^H\mathbf{p}_m|^2$, and $\mathbf{f}_k=\begin{bmatrix} f_{1,k}&\cdots f_{M,k}\end{bmatrix}^T$.

 The aim of this letter is to maximize the far-field users' sum data rate while guaranteeing the near-field users' QoS requirements, based on the following optimization problem:    
 \begin{problem}\label{pb:1} 
  \begin{alignat}{2}
\underset{P_m, f_{m,k} }{\rm{max}} &\quad     \sum^{K}_{k=1}
\min\left\{ R_k^{\rm FF},R_{m,k}^{\rm FF-NF}, m\in\mathcal{S}_k  \right\}  \label{1tst:1}
\\ s.t. &\quad R_m^{\rm NF} \geq R, \sum^{K}_{k=1} \mathbf{1}_{ f_{m,k}\neq 0}\leq 1,1\leq m \leq M  \label{1tst:2} \\
& \quad  P_m+  \sum^{K}_{k=1}|f_{m,k}|^2 \leq P, 1\leq m \leq M   \label{1tst:3}\\
&\quad P_m\geq 0, 1\leq m \leq M, \label{1tst:4}
  \end{alignat}
\end{problem} 
where $ \mathbf{1}_{x\neq0}$ denotes an indicator function, i.e., $ \mathbf{1}_{x\neq0}=1$ if $x\neq0$, otherwise $ \mathbf{1}_{x\neq0}=0$,  $R$ denotes the near-field users' target data rate, and $P$ denotes the transmit budget per beam.  We note that the resource allocation problem    in \eqref{pb:1} requires the base station to have access to the users' channel state information (CSI). This CSI assumption    can be realized by  asking each user to first carry out channel estimation based on the pilots broadcast by the base station and then feed back its CSI to the base station via a reliable control channel.

{\it Remark 2:} The objective function in \eqref{1tst:1} indicates that   the $k$-th far-field user prefers to use a beam on which both $h_m$ and $g_{m,k}$ are strong. Therefore, to find $\mathcal{S}_k$ and remove the indicator function in  \eqref{1tst:1},  a simple sub-optimal scheduling scheme can be used first,    where the far-field users are successively asked    to select the best $D_x$ beams based on the following criterion: ${\arg}~\underset{m}{ \max}~ \min\left\{\frac{h_m}{\max\{h_1, \cdots,h_M\}}, \frac{g_{m,k}}{\max\{g_{1,k},\cdots,g_{M,k}\}} \right\}$. Here, the channel gains, $h_m$ and $g_{m,k}$, are normalized to ensure that they are in the same order of magnitude. As a result,   $f_{m,k}=0$ for $m\not \in \mathcal{S}_k $, and only $f_{m,k}$,  $m  \in \mathcal{S}_k $, need to be optimized.   We note that    optimal scheduling is possible   by applying  integer programming  as in \cite{beamallocation}. However, this  is out of the scope of this paper due to the space limitations.
\vspace{-1em}

\section{Proposed Resource Allocation Algorithms} \vspace{-0em}
\subsection{SCA-based  Resource Allocation}
In this section, the general case with $K\geq 1$ and $D_x\geq 1$ is considered first. Problem \eqref{pb:1}  can be first recast as follows: 
{\small  \begin{problem}\label{pb:2} 
  \begin{alignat}{2}
\underset{P_m,x_k, f_{m,k} }{\rm{max}}  &\quad     \sum^{K}_{k=1}
\log(1+x_k) \label{2st:0}
\\ s.t. &\quad  
 | \tilde{\mathbf{g}}^H_k\mathbf{f}_k|^2 \gamma_k^{-1} \geq x_k\geq0,  \quad 1\leq k \leq K\label{2st:1}
\\&\quad \frac{  |f_{m,k}|^2 h_m }{\sigma^2+P_mh_m} \geq x_k, m\in\mathcal{S}_k,1\leq k \leq K   \label{2st:2}
\\ &\quad f_{m,k}=0, m\not\in \mathcal{S}_k, 1\leq k \leq K, \label{2st:00}
\\ &\quad P_m \geq \frac{\sigma^2\epsilon}{h_m}, 1\leq m \leq M, \label{2st:3} \\
& \quad  \eqref{1tst:3}, \eqref{1tst:4},
  \end{alignat}
\end{problem} }
\hspace{-0.5em}where $\epsilon=2^R-1$, and constraint \eqref{2st:00} is due to the use of the scheduling scheme discussed in   Remark 3.  
Because \eqref{2st:1} and \eqref{2st:2} are decreasing functions of $P_m$,  it is straightforward to show that the optimal solution of $P_m$ is  $P_m^*=\min\left\{\frac{\sigma^2\epsilon}{h_m},P \right\}$. We   note that the main challenges involving   problem \eqref{pb:2} are  the non-convex constraints in \eqref{2st:1} and \eqref{2st:2}, which motivates the use of SCA. To facilitate the application of SCA, problem \eqref{pb:2} can be first recast as follows: 
{\small \begin{problem}\label{pb:4} 
  \begin{alignat}{2}
\underset{ x_k, f_{m,k} }{\rm{max}}  &\quad     \sum^{K}_{k=1}
\log(1+x_k) \label{4st:0}
\\ s.t. &\quad  \eta_k +\underset{i\neq k}{\sum} | \tilde{\mathbf{g}}^H_k\mathbf{f}_i|^2\leq 
 \frac{ | \tilde{\mathbf{g}}^H_k\mathbf{f}_k|^2}
 {x_k}  ,1\leq k \leq K\label{4st:1}
\\&\quad      x_k \mu_m\leq  |f_{m,k}|^2, m\in\mathcal{S}_k, 1\leq k \leq K    \label{4st:2} \\
& \quad    \sum^{K}_{k=1}|f_{m,k}|^2 \leq P-P_m^*, 1\leq m \leq M, x_k\geq 0   \label{3st:3}\\&\quad  \eqref{2st:00}\nonumber,
  \end{alignat}
\end{problem} }
\hspace{-0.5em}where $\eta_k=\sigma^2+\sum^{M}_{m=1}P_m^* g_{m,k} $ and $\mu_m =\frac{\sigma^2+P_m^*h_m}{h_m}, m\in\mathcal{S}_k $. The term
$ \frac{ | \tilde{\mathbf{g}}^H_k\mathbf{f}_k|^2}
 {x_k}$ can be approximated as an affine function  via the Taylor expansion. In particular,  first express $ | \tilde{\mathbf{g}}^H_k\mathbf{f}_k|^2$ as follows:
 \begin{align}
  | \tilde{\mathbf{g}}^H_k\mathbf{f}_k|^2 = \bar{\mathbf{f}}_k^T
  \left(\hat{\mathbf{g}}_k\hat{\mathbf{g}}_k^T+\check{\mathbf{g}}_k\check{\mathbf{g}}_k^T\right) \bar{\mathbf{f}}_k,
 \end{align}
 where $(\cdot)^T$ denote the transpose,   $ \bar{\mathbf{f}}_k=\begin{bmatrix}{\rm Re}(  {\mathbf{f}}_k)^T &{\rm Im}(  {\mathbf{f}}_k)^T\end{bmatrix}^T$, $\hat{\mathbf{g}}_k=\begin{bmatrix}{\rm Re}(\hat{\mathbf{g}}_k)^T &{\rm Im}(\hat{\mathbf{g}}_k)^T\end{bmatrix}^T$, $\check{\mathbf{g}}_k=\begin{bmatrix}-{\rm Im}(\hat{\mathbf{g}}_k)^T &{\rm Re}(\hat{\mathbf{g}}_k)^T\end{bmatrix}^T$. By  building the   real-valued vector, $\tilde{\mathbf{f}}_k=\begin{bmatrix}
  \bar{\mathbf{f}}_k^T& x_k
  \end{bmatrix}^T$, and applying   the first order Taylor expansion, $ \frac{ | \tilde{\mathbf{g}}^H_k\mathbf{f}_k|^2}
 {x_k}$ can be approximated as follows: 
\begin{align}
 \frac{ | \tilde{\mathbf{g}}^H_k\mathbf{f}_k|^2}
 {x_k} \approx  \frac{ | \tilde{\mathbf{g}}^H_k\mathbf{f}_k^0|^2}
 {x_k^0} + 
 \triangledown^T_k|_{\tilde{\mathbf{f}}_k=\tilde{\mathbf{f}}_k^0}  \left(\tilde{\mathbf{f}}_k-\tilde{\mathbf{f}}^0_k \right),
\end{align}
where $\tilde{\mathbf{f}}^0_k=\begin{bmatrix}{\rm Re}( {\mathbf{f}}_k^0)^T &{\rm Im}(  {\mathbf{f}}_k^0)^T&x_k^0\end{bmatrix}^T $ denotes the initial value  of $\tilde{\mathbf{f}}_k$, and   $\triangledown_k$ is given by
\begin{align}
\triangledown_k = \begin{bmatrix}
 \frac{2 \bar{\mathbf{f}}_k ^T\left(\hat{\mathbf{g}}_k\hat{\mathbf{g}}_k^T+\check{\mathbf{g}}_k\check{\mathbf{g}}_k^T\right) ^T}
 {x_k} & -\frac{ | \tilde{\mathbf{g}}^H_k\mathbf{f}_k|^2}
 {x_k^2}
\end{bmatrix}^T.
  \end{align}  
Constraint \eqref{4st:2} can be similarly approximated by using an initial value $f_{m,k}^0$.

Therefore,  problem \eqref{pb:4} can be approximated as follows: 
  {\small  \begin{problem}\label{pb:5} 
  \begin{alignat}{2}
\underset{ x_k, f_{m,k} }{\rm{max}} &\quad     \sum^{K}_{k=1}
\log(1+x_k) \label{5st:0}
\\ s.t. &\quad  \eta_k +\underset{i\neq k}{\sum} | \tilde{\mathbf{g}}^H_k\mathbf{f}_i|^2\leq 
 \frac{ | \tilde{\mathbf{g}}^H_k\mathbf{f}_k^0|^2}
 {x_k^0} + 
 \triangledown^T_k|_{\tilde{\mathbf{f}}_k=\tilde{\mathbf{f}}_k^0}  \left(\tilde{\mathbf{f}}_k-\tilde{\mathbf{f}}^0_k \right),\nonumber \\
 &\quad    \quad\quad\quad \quad\quad\quad\quad\quad\quad\quad\quad\quad\quad 1\leq k \leq K\label{5st:1}
\\&\quad      x_k \mu_m\leq  |f^0_{m,k}|^2 +4{\rm Re}\left\{ (f_{m,k}^0)^H (f_{m,k}-f^0_{m,k})\right\}, \nonumber\\
&\quad\quad\quad\quad\quad  \quad\quad\quad\quad\quad\quad m\in\mathcal{S}_k , 1\leq k \leq K   \label{5st:2} \\
& \quad   \eqref{2st:00}, \eqref{3st:3}\nonumber,
  \end{alignat}
\end{problem}} 
\hspace{-0.5em}which is   a convex optimization problem and can be straightforwardly solved by convex solvers.

The implementation of SCA requires that  $\tilde{\mathbf{f}}_k^0$ is a feasible solution of problem \eqref{pb:4}, and $\tilde{\mathbf{f}}_k^0$  can be obtained as follows. First,  by using  \eqref{3st:3}, we choose $f_{m,k}^0=0$ for $m\not\in \mathcal{S}_k$, and  $f_{m,k}^0=(P-P^*_m)  
 $ for $m\in \mathcal{S}_k$, which means that the following choices of $x_k^0$, $1\leq k \leq K$, are feasible: 
\begin{align}
\label{greedy}
x_k^0=\min\left\{\frac{ | \tilde{\mathbf{g}}^H_k\mathbf{f}_k^0|^2}{ \eta_k +\underset{i\neq k}{\sum} | \tilde{\mathbf{g}}^H_k\mathbf{f}_i^0|^2}, \frac{ |f_{m,k}^0|^2}{\mu_m}, m\in\mathcal{S}_k 
\right\}.
\end{align}  Based on $\tilde{\mathbf{f}}_k^0$, SCA can be applied  in an iterative manner, i.e., by using $\tilde{\mathbf{f}}_k^0$ and solving problem \eqref{pb:5}, a new estimate of $\tilde{\mathbf{f}}_k$ can be generated and used to replace $\tilde{\mathbf{f}}_k^0$ in the next iteration. In general, SCA can only converge to a stationary point, which cannot be guaranteed to be the optimal solution. Motivated by this, the optimal performance of NOMA transmission is studied for the two special cases in the following subsection.

   \vspace{-1em}
  
  \subsection{Optimal Performance in Two Special Cases}
  \subsubsection{Special case with $K=1$}\label{subsection1} If there is a single far-field user, i.e., $K=1$, problem \eqref{pb:1} can be simplified as follows:\footnote{The subscript, $k$, is omitted since there is a single far-field user.}
 \begin{problem}\label{pb:6} 
  \begin{alignat}{2}
\underset{ x, f_{m} }{\rm{max}} &\quad    x
 \label{6tst:1}
\\ s.t. &\quad 
R^{\rm FF}\geq x, x\geq0,R_m^{\rm NF} \geq R, 1\leq m \leq M \label{6st:1} \\
& \quad R_m^{\rm FF-NF}\geq x,  m\in\mathcal{S}, f_m=0, m\not\in \mathcal{S} \label{6st:2}\\
& \quad    P_m+  |f_m|^2 \leq P,  P_m\geq 0, 1\leq m \leq M\label{6st:4}. 
  \end{alignat}
\end{problem} 

By following   steps similar to those in the previous subsection, problem \eqref{pb:6} can be recast as follows:
 {  \begin{problem}\label{pb:7} 
  \begin{alignat}{2}
\underset{ y, f_{m} }{\rm{max}}&\quad    y
 \label{7tst:1}
\\ s.t. &\quad 
 | \mathbf{g}^H\mathbf{P}\mathbf{f}|^2 
 \geq \eta_0 y,  y\geq 0\label{7st:2} \\
& \quad  f_m^2 h_m  \geq \eta_m y, m\in \mathcal{S},f_m=0, m\not\in \mathcal{S} \label{7st:3}\\ &\quad  |f_m|^2\leq P-  P_m^* , 1\leq m \leq M ,   \label{7st:4}
  \end{alignat}
\end{problem}} 
\hspace{-0.5em}where $y=2^x-1$, $P_m^* = \frac{\sigma^2\epsilon}{h_m}$, $\eta_0=\sigma^2+\sum^{M}_{m=1}P_m^* g_m$, and  $ \eta_m=\sigma^2+P_m^*h_m$. While problem \eqref{pb:7} is not convex, it can be solved by using the following lemma.

\begin{lemma}\label{lemma1}
The optimal value  of problem   \eqref{pb:7}  is the same as that of the following optimization problem: 
 {\small \begin{problem}\label{pb:8} 
  \begin{alignat}{2}
\underset{ y, f_{m} }{\rm{max}}&\quad    y
 \label{8tst:1}
\\ s.t. &\quad 
 \left(\sum^{M}_{m=1} g_m^{\frac{1}{2}} |f_m|\right)^2 
 \geq \eta_0 y, y\geq 0, \label{8tst:2}   \eqref{7st:3},  \eqref{7st:4}, 
  \end{alignat}
\end{problem} }
\hspace{-0.5em}where $g_m = |\mathbf{g}^H\mathbf{p}_m|^2$. 
\end{lemma}
\begin{proof}
Denote a feasible solution of problem \eqref{pb:7} by $(y^0, \mathbf{f}^0\triangleq \begin{bmatrix} f_1^0 & \cdots &f_M^0\end{bmatrix}^T)$. Constraint \eqref{7st:2} ensures that $| \mathbf{g}^H\mathbf{P}\mathbf{f}^0|^2 
 \geq \eta_0 y^0$.  Because  $\sum^{M}_{m=1} g_m^{\frac{1}{2}} |f_m^0|\geq | \mathbf{g}^H\mathbf{P}\mathbf{f}^0| $, $ \left(\sum^{M}_{m=1} g_m^{\frac{1}{2}} |f_m^0|\right)^2 
 \geq \eta_0 y^0$, which means that any feasible solution of problem   \eqref{pb:7} is also feasible for problem \eqref{pb:8}. In other words, the feasible set of problem \eqref{pb:7} is a subset of that  of problem \eqref{pb:8}, and  hence the optimal value of     problem \eqref{pb:8} is no less than that of   problem of \eqref{pb:7}.  We also note that   the optimal solution of  problem \eqref{pb:8}  leads to    a feasible solution of problem \eqref{pb:7}. For example,   assume that $(y^*, \mathbf{f}^*\triangleq \begin{bmatrix} f_1^* & \cdots &f_M^*\end{bmatrix}^T)$ is the optimal solution of  problem \eqref{pb:8}.   $(y^*,\tilde{\mathbf{f}}^* \triangleq  \begin{bmatrix}|f_1^*|e^{-j\bar{\theta}_1}&\cdots &|f_M^*|e^{-j\bar{\theta}_M} \end{bmatrix}^T)$  must be feasible to   problem \eqref{pb:7}, where $\bar{\theta}_m$ is the argument of the complex-valued number $\mathbf{g}^H\mathbf{p}_m$. By using the fact that     the optimal value of     problem \eqref{pb:8} is no less than that of   problem of \eqref{pb:7}, $y^*$ must be the optimal value of problem \eqref{pb:7} as well. Therefore, problems \eqref{pb:7} and \eqref{pb:8} have the same optimal value, and the proof is complete. 
\end{proof}
By using Lemma \ref{lemma1}, the optimal performance of NOMA transmission in the special case with $K=1$   can be obtained by defining  $z_m = |f_m|^2$ and transferring problem \eqref{pb:8}   into the following equivalent convex form:
{\small  \begin{problem}\label{pb:9} 
  \begin{alignat}{2}
\underset{ y, z_{m} }{\rm{max}}&\quad    y
 \label{9tst:1}
\\ s.t. &\quad 
  \left(\sum^{M}_{m=1} g_m^{\frac{1}{2}} z_m^{\frac{1}{2}}\right)^2 
 \geq \eta_0 y, y\geq 0, z_m h_m  \geq \eta_m y, m\in\mathcal{S}\nonumber \\
& \quad  0\leq z_m\leq P-  P_m^* , 1\leq m \leq M, z_m=0, m\not\in \mathcal{S}    .\nonumber
  \end{alignat}
\end{problem} }


 \begin{algorithm}[t]
\caption{Branch and Bound Algorithm}

 \begin{algorithmic}[1] 
 
\State Set $\bar{\mathcal{S}}_0=\{\mathcal{B}_0\}$ and   tolerance $\epsilon$,  $i=0$,  $\beta^u_0=\phi^{\rm up}(\mathcal{B}_0)$, $\beta^l_0=\phi^{\rm lb}(\mathcal{B}_0)$, and $\delta=\beta^u_0-\beta^l_0$
\While { $\delta\geq \epsilon$ }
\State $i=i+1$
\State  Find $\mathcal{B}\in\bar{\mathcal{S}}_{i-1}$ with the criterion: $\min \phi^{\rm lb}(\mathcal{B})$

\State Split $\mathcal{B}$ along its longest edge into $\mathcal{B}_{1}$ and $\mathcal{B}_2$

\State Construct $\bar{\mathcal{S}}_i=\{\mathcal{B}_1\cup \mathcal{B}_2 \cup (\bar{\mathcal{S}}_{i-1}\backslash \mathcal{B}) \}$

\State Update the   upper and lower bounds $\beta^u_i=\max \phi^{\rm up}(\mathcal{B})$ and $\beta^l_i=\max \phi^{\rm lb}(\mathcal{D})$, $\forall \mathcal{B}\in \bar{\mathcal{S}}_i$.  

\State  $\delta=\beta^u_i-\beta^l_i$

\State Prune    $\mathcal{B}$ with upper  bounds   smaller than $\beta^l_i$.

 \EndWhile
\State \textbf{end}
 \end{algorithmic}\label{algorithm}
\end{algorithm}

  \subsubsection{Special case with $D_x=1$}\label{subsection2} If each far-field user uses a single beam,  problem \eqref{pb:1} can be simplified as follows: 
 {  \begin{problem}\label{pb:10} 
  \begin{alignat}{2}
\underset{ x_k, f_{k}}{\rm{max}} &\quad     \sum^{K}_{k=1}
\log(1+x_k) \label{10st:0}
\\ s.t. &\quad  \eta_k +\underset{i\neq k}{\sum} | \tilde{ {g}}_k|^2 | {f}_i|^2\leq 
 \frac{ | \tilde{ {g}}_k|^2| {f}_k|^2}
 {x_k}  , 1\leq k \leq K \label{10st:1} \\&\quad   |f_{k}|^2  \geq x_{k} \frac{\sigma^2+P_{m_k}^*h_{m_k}}{h_{m_k}} , 1\leq k\leq K  \label{10st:2} \\
& \quad    |f_{k}|^2 \leq P-P_{m_k}^*, 1\leq m \leq M    \label{10st:3},
  \end{alignat}
\end{problem} }
\hspace{-0.5em}where $m_k$ denotes the index of the beam used by the $k$-th far-field user, and $f_{m,k}$ is simplified to $f_k$ for this special case. 

  \begin{figure}[t] \vspace{-1em}
\begin{center}
\subfigure[$N=64$ ]{\label{fig1a}\includegraphics[width=0.3\textwidth]{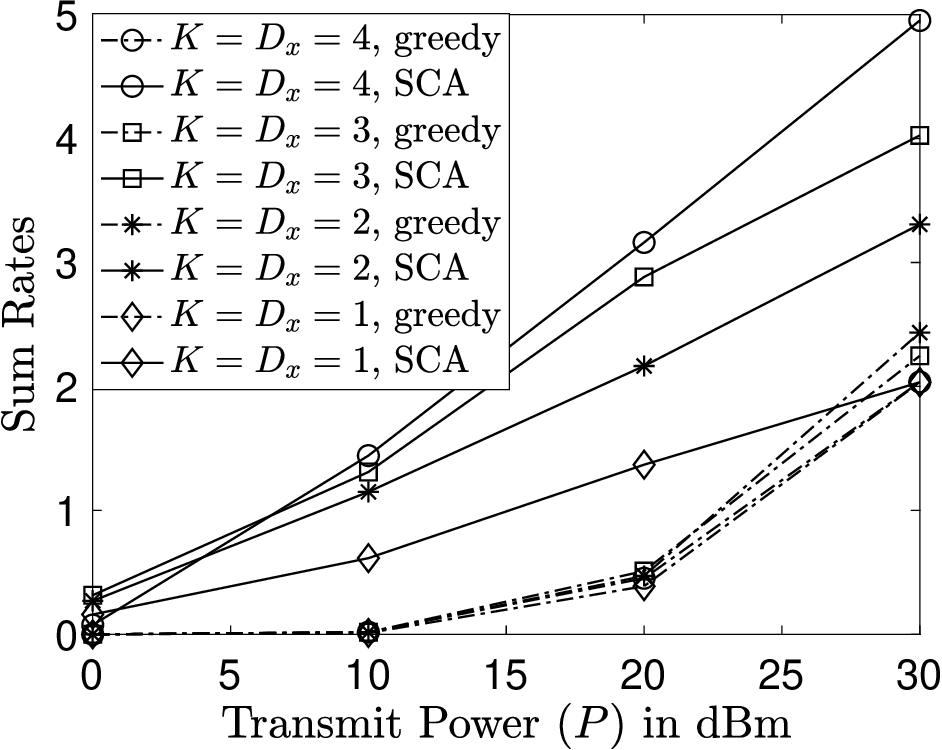}}\hspace{2em}
\subfigure[$N=128$]{\label{fig1b}\includegraphics[width=0.3\textwidth]{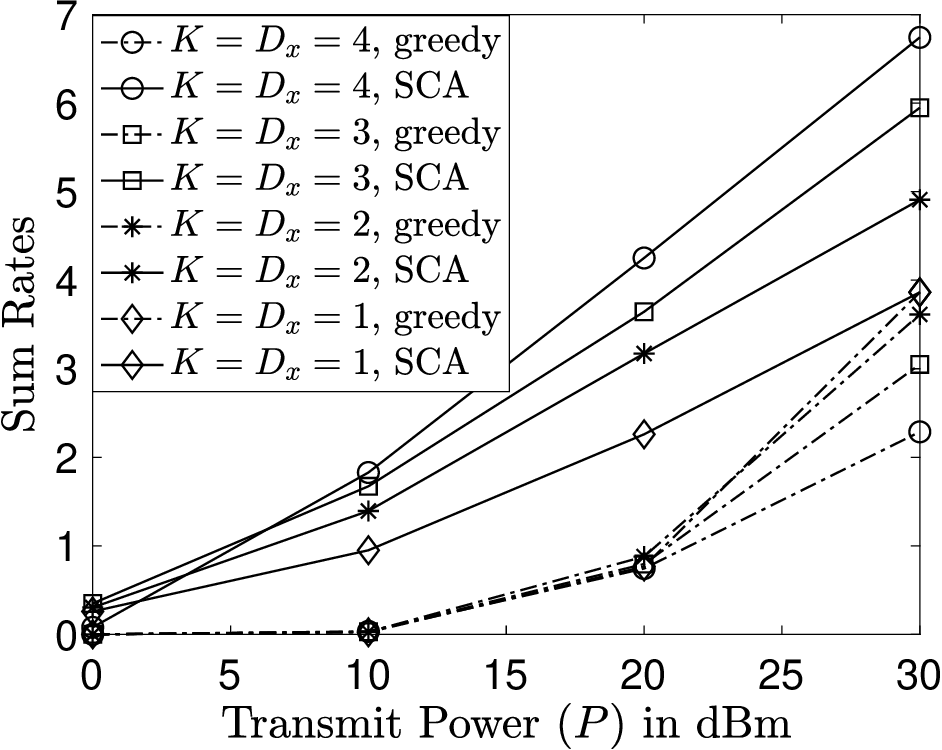}} \vspace{-1em}
\end{center}
\caption{ Far-field users' sum data rates achieved by NOMA with randomly located users. The greedy benchmarking scheme is  based on \eqref{greedy}.     \vspace{-1em} }\label{fig1}\vspace{-1em}
\end{figure}

For this special case,   problem \eqref{pb:10} is similar to conventional power allocation in interference channels, where the BB algorithm can be used to obtain the optimal solution \cite{5765556, beamallocation}. 
Due to the space limitations, the  principle of the BB algorithm is described only briefly in the following. As shown in Algorithm \ref{algorithm},  the initilization of the algorithm   builds   an initial   box,  denoted by $\mathcal{B}_0$, by using an upper bound on $x_k$ as follows: $
x_k \leq  \min\left\{ \frac{| \tilde{ {g}}_k|^2(P-P^*_m) }{\eta_k},
\frac{P-P^*_m}{\mu_m}, m\in \mathcal{S}_k
\right\}$.
 
In each iteration of the BB algorithm, the key step is to calculate the upper and lower bounds for a  box,  $\mathcal{B}$, which are denoted by $ \phi^{\rm up}(\mathcal{B})$ and $ \phi^{\rm lb}(\mathcal{B})$, respectively. 
Further denote  the minimum and maximum  vertices of    $\mathcal{B}$ by $\mathbf{x}_{\max}$ and $\mathbf{x}_{\min}$, respectively.   $ \phi^{\rm up}(\mathcal{B})=\sum^{K}_{k=1}
\log(1+x_{k,\max}) $ and $ \phi^{\rm lb}(\mathcal{B})\sum^{K}_{k=1}
\log(1+x_{k,\min}) $, if  $\mathbf{x}_{\min}$ is a   solution of the following feasibility optimization  problem:
 \begin{problem}\label{pb:11} 
  \begin{alignat}{2}
\underset{   f_{k} }{\rm{max}}&\quad    1 \label{11st:0}
\\ s.t. &\quad  \eta_k +\underset{i\neq k}{\sum} | \tilde{ {g}}_k|^2 | {f}_i|^2\leq 
 \frac{ | \tilde{ {g}}_k|^2| {f}_k|^2}
 {x_{k,\min}}  , 1\leq k \leq K \label{10st:1} \\&\quad   |f_{k}|^2  \geq x_{k,\min} \frac{\sigma^2+P_{m_k}^*h_{m_k}}{h_{m_k}} , 1\leq k\leq K  \label{11st:2} \\
& \quad    |f_{k}|^2 \leq P-P_{m_k}^*, 1\leq m \leq M    \label{11st:3},
  \end{alignat}
\end{problem} 
where $x_{k,\max}$ and $x_{k,\min}$ are the $k$-th elements of $\mathbf{x}_{\max}$ and $\mathbf{x}_{\min}$, respectively.
If $\mathbf{x}_{\min}$ is not feasible, the upper and lower bounds are set as $0$. The details for implementing the BB algorithm can be found in \cite{beamallocation}. 
\vspace{-1em}
  \section{Simulation Results}
In this section, simulation results are presented to evaluate the performance of the proposed NOMA scheme. For all  simulations we used, $f_c=28$ GHz, $d=\frac{\lambda}{2}$,  $R=0.1$ bits per channel use, $\sigma^2=-80$ dBm,   and  $M=36$ \cite{9738442}. The ULA is placed on the vertical coordinate  axis and ${\boldsymbol \psi}_0=(0,0)$.

In Fig. \ref{fig1}, the performance of NOMA transmission is evaluated with   randomly located users. In particular, the near-field users are uniformly located inside of a half-ring  with its inner and outer radii being $5$ m and $d_{R}(64)$ m, respectively.   The far-field users are uniformly located inside of a half-ring with its inner and outer radii being $d_{R}(128)$ m and $(d_{R}(128)+10)$ m, respectively.  A greedy resource allocation scheme based on \eqref{greedy} is used as a benchmarking scheme in the figure.  As can be seen from  Fig. \ref{fig1}, the use of NOMA   ensures that the  spatial beams preconfigured for the near-field users  are   efficiently utilized  to support additional far-field users.     Comparing Fig. \ref{fig1a} to Fig. \ref{fig1b}, one can also observe that the use of more antennas at the base station can further improve the performance of NOMA, which indicates the importance of NOMA for massive MIMO. 
  \begin{figure}[t]\centering \vspace{-2.5em}
    \epsfig{file=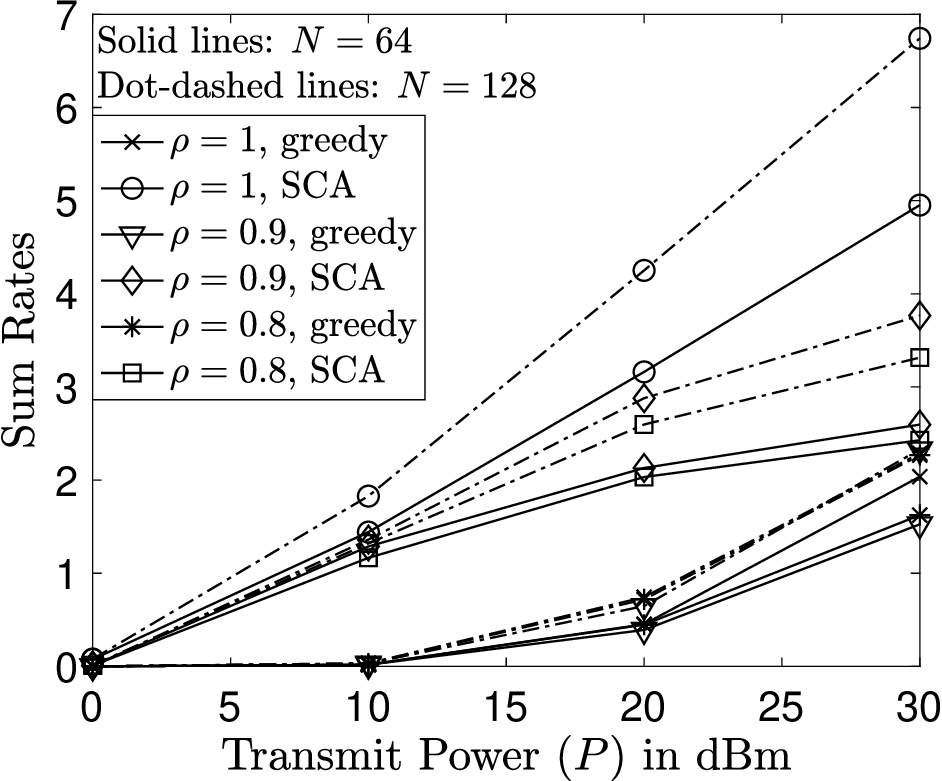, width=0.33\textwidth, clip=}\vspace{-0.5em}
\caption{ Illustration of the impact of imperfect CSI. $K=D_x=4$.   \vspace{-1em}    }\label{fig6}   \vspace{-2em} 
\end{figure}
 In Fig. \ref{fig6}, the impact of imperfect CSI on NOMA is studied. We assume that the perfect CSI of the legacy near-field users  is   available at the base station, but there exist CSI  errors for the far-field users. The estimated CSI is modelled as  $\hat{\mathbf{g}}_k=\rho\mathbf{g}_k+ \sqrt{1-\rho}\mathbf{e}_k$, where   $\mathbf{e}_k$ denotes the CSI errors     due to imperfect CSI feedback  or     abrupt changes in    the users' channels, $\rho\in[0,1]$ denotes a parameter to measure the quality of the CSI, and   $\mathbf{e}_k$ follows a complex Gaussian distribution with zero mean and variance $\sigma^2_e=\alpha_k^2$ \cite{8901184}.   Fig. \ref{fig6} shows that  the performance of NOMA is degraded   by  imperfect CSI, particularly for large $\rho$. Fig. \ref{fig6} also shows that the proposed  SCA scheme    still outperforms the baseline, even in the presence of imperfect CSI.  We note that robust beamforming can be used to combat the detrimental effects of  imperfect CSI, which is beyond the scope of this paper.  

As shown in Sections \ref{subsection1} and \ref{subsection2}, the optimal performance of NOMA transmission can be obtained for the two special cases, which are studied in Fig. \ref{fig2} by focusing the following  deterministic case.   On the one hand, assume that  the $M=36$ near-field users are equally spaced with $\frac{10}{\sqrt{M}}$ m distance, and located within a square with its center  located at $(0,9)$ m.  On the other hand, assume that the far-field users are also equally spaced and located on a half-circle with radius $90$ meters.   Fig. \ref{fig2a} focuses on the scenario with    a single scheduled far-field user, and   Fig. \ref{fig2b} focuses on the scenario with a single available beam. Fig. \ref{fig2}   shows that the far-field user's performance can be improved by using more beams, and there is an optimal choice of $K$ for   sum-rate maximization.   Fig. \ref{fig2} also shows that       SCA   provides a reasonable estimate for the optimal performance.
  \begin{figure}[t] \vspace{-2.5em}
\begin{center}
\subfigure[$K=1$ ]{\label{fig2a}\includegraphics[width=0.33\textwidth]{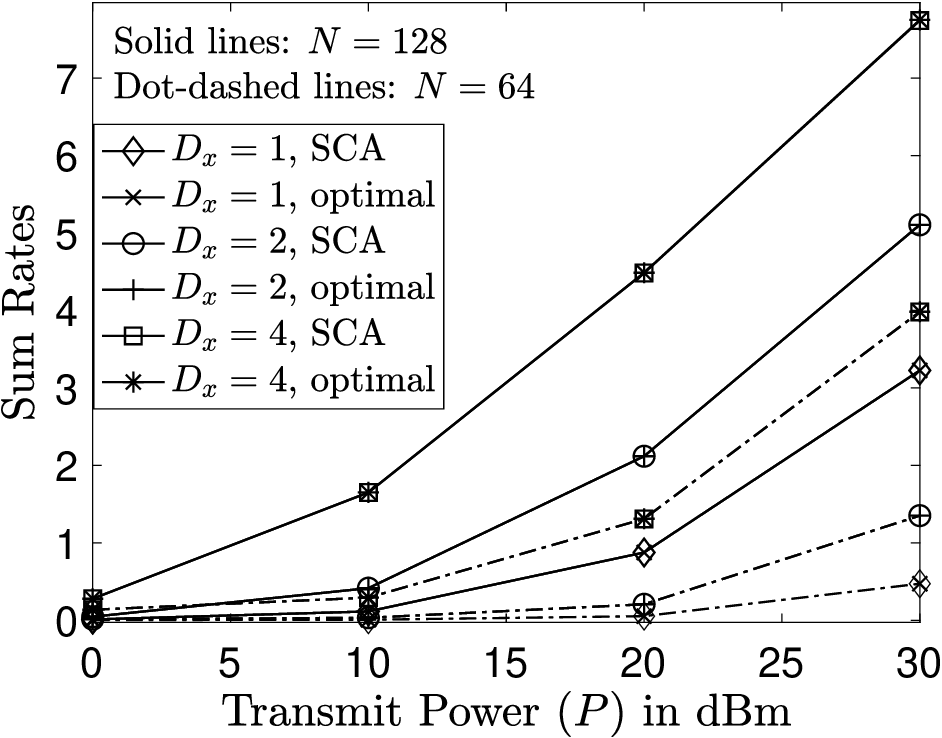}}\hspace{2em}
\subfigure[$D_x=1$]{\label{fig2b}\includegraphics[width=0.33\textwidth]{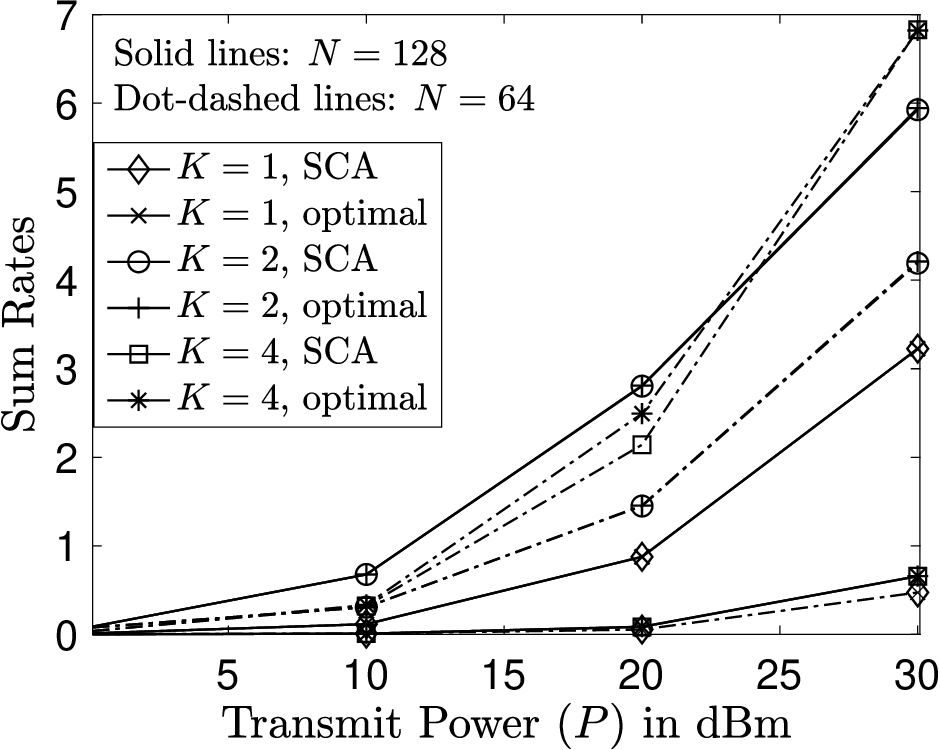}} \vspace{-1em}
\end{center}
\caption{ Deterministic studies for the optimality    of         NOMA transmission.        \vspace{-1em} }\label{fig2}\vspace{-2em}
\end{figure}\vspace{-1.5em}
\section{Conclusions}
This letter has considered  a legacy network, where spatial beams have been preconfigured for legacy near-field users, and shown  that via NOMA,   additional  far-field users can be efficiently served by  using these preconfigured    beams.  In this letter, each user was assumed to have a single antenna. An important direction for future research is to study the design of NOMA assisted transmission for users equipped with multiple antennas. Furthermore, perfect SIC has been assumed. The study of the impact of imperfect SIC on NOMA constitutes   another important direction for future research.  \vspace{-2em}
\bibliographystyle{IEEEtran}
\bibliography{IEEEfull,trasfer}
  \end{document}